\documentclass{birkjour}

\usepackage{comment}
\usepackage{graphics}
\usepackage{slashed}
\usepackage{amssymb}
\usepackage{lscape}
\usepackage{amsmath}
\usepackage{rotating}
\usepackage{graphicx}
\usepackage{amsthm}
\usepackage{MnSymbol}
\usepackage{tensor}
\usepackage{bbm}
\usepackage{tikz}
\usepackage{tikz-cd}
\usepackage{scalerel}

\usepackage[noadjust]{cite}

\usepackage{fancyvrb}
\usepackage{undertilde}

\usepackage{hyperref}
\newcommand{\cl}{C \kern -0.1em \ell}

\newcommand{\opn}{\operatorname}
\newcommand{\clpq}{\cl_{p,q}}
\newcommand{\clpqr}{\cl_{p,q}(\Bbb{R})}
\newcommand{\clpqc}{\cl_{p,q}(\Bbb{C})}
\newcommand{\clunotres}{\cl_{1,3}}
\newcommand{\cldostres}{\cl_{2,3}}
\newcommand{\clunotresr}{\cl_{1,3}(\Bbb{R})}
\newcommand{\clunotresc}{\cl_{1,3}(\Bbb{C})}
\newcommand{\cldostresr}{\cl_{2,3}(\Bbb{R})}

\newcommand{\pinunotres}{\opn{Pin}({1,3})}

\newcommand{\spindostres}{\opn{Spin}({2,3})}

\newcommand{\Ad}{\opn{Ad}}

\newcommand{\R}{\Bbb{R}}

\newcommand{\K}{\Bbb{K}}

\newcommand{\ub}{\mathbf{u}}

\theoremstyle{plain}
\newtheorem{theo}{Theorem}

\usetikzlibrary{calc}

\theoremstyle{definition}
\newtheorem{defn}{Definition}

\newtheorem{obs}{Observation}

\begin{document}
\title[Complexifying the STA by means of an extra timelike dimension]{Complexifying the spacetime algebra by means of an extra timelike dimension: Pin, Spin and algebraic spinors}
\author[M.R.A. Arcodía]{Marcos R. A. Arcod\'{\i}a}
\address{%
Instituto de Astronom\'{i}a y F\'{i}sica del Espacio (CONICET \& UBA), Casilla de Correo 67, Sucursal 28, 1428 Buenos Aires, Argentina}
\email{marcodia@iafe.uba.ar}
\begin{abstract}
Because of the isomorphism $\clunotresc\cong\cldostresr$, it is possible to complexify the spacetime Clifford algebra $\clunotresr$ by adding one additional timelike dimension to the Minkowski spacetime. In a recent work we showed how this treatment provide a particular interpretation of Dirac particles and antiparticles in terms of the new temporal dimension. In this article we thoroughly study the structure of the real Clifford algebra $\cldostresr$ paying special attention to the isomorphism $\clunotresc\cong\cldostresr$ and the embedding $\clunotresr\subseteq\cldostresr$. On the first half of this article we analyze the Pin and Spin groups and construct an injective mapping $\opn{Pin}(1,3)\hookrightarrow\opn{Spin}(2,3)$, obtaining in particular elements in $\opn{Spin}(2,3)$ that represent parity and time reversal. On the second half of this paper we study the spinor space of the algebra and prove that the usual structure of complex spinors in $\clunotresc$ is reproduced by the Clifford conjugation inner product for real spinors in $\cldostresr$.
\end{abstract}
\maketitle
\sloppy
\section{Introduction}

The use of Clifford algebra (CA) provides formalization and foundation when working with spinor fields in physics. For instance, the fact that in ordinary quantum mechanics, the state of a $1/2$-spin particle is represented by a two-component complex vector in a Hilbert space with a positive definite metric can be inferred by constructing the Clifford algebra of $3D$ space. As a matter of fact, this Clifford algebra is isomorphic to $\mathcal{M}(2,\mathbb{C})$, and consequently the representation space for its irreducible representation is $\mathbb{C}^{2}$. However, there are many definitions of spinors --- the one used above is that of a real algebraic spinor.

It is also possible to consider the state to be represented by a real classical spinor, this is, a representation of the group $\opn{Spin}(3)$, in which case we would obtain that the state is given by a single quaternion. If we complexify the classical spinor, we obtain back the algebraic spinor in $\mathbb{C}^{2}$ (but as semi-spinors instead). Nonetheless, not in every case the complexified classical spinors are straightforwardly related to real algebraic spinors. Consider for instance, the spacetime algebra: this is the CA of the Minkowski space-time with signature $(+---)$. A real algebraic spinor in this algebra is represented by an element in $\mathbb{H}^{2}$, while a complex classical spinor is an element in $\mathbb{C}^{2}\oplus\mathbb{C}^{2}$, and a complex algebraic spinor is an element in $\mathbb{C}^{4}$. In the light of these observations, a question we may ask is what is the correct structure of a physical spinor space? Should we require it to be a representation just of a certain group, or of the whole Clifford algebra?

In the context of the Dirac equation for the mentioned Minkowski spacetime:
\begin{equation}\label{diraceq}
\gamma_{\mu}\partial^{\mu}\psi+i\frac{mc}{\hbar}\psi=0 \quad,
\end{equation}
this question has a simple answer: since the generators of the algebra, $\gamma_{\mu}$, must act on the spinor $\psi$ then, we have to require that whatever space $\psi$ lives in, it has to carry a representation of the whole Clifford algebra. Furthermore, since the mass term is multiplied by the imaginary unit $i$, we can conclude that the appropriate representation space for $\psi$ is that of the complexified Clifford algebra representation.

Of course, this answer may be indeed too simple: as was noted by Hestenes in his foundational papers\cite{hestenes1,hestenes2,hestenes3}, it is possible to work only with the real Clifford algebra by replacing algebraic spinors with \emph{operator spinors}\footnote{see chapter 6 of \cite{VdR} for terminology.} and the Dirac equation with the Dirac-Hestenes equation. However, if we want to recover the algebraic spinor appearing in Eq. \eqref{diraceq}, it is necessary to work in the complex spacetime algebra\footnote{The idempotent $u_{1}$ appearing in \cite{hestenes3} can only be constructed in the complex case: in the real spacetime algebra such an element doesn't exist.}.

If we accept to work with the algebraic spinors appearing in the Dirac equation \eqref{diraceq}, it is necessary to consider the complex Clifford algebra $\clunotresc$. Regarding this aspect, a known fact about $\clunotresc$ is that it is isomorphic to the real Clifford algebra $\cldostresr$. This isomorphism provides an equivalence between complexifying the Clifford algebra $\clunotresr$ and adding an extra timelike dimension to the spacetime. Furthermore, this isomorphism also provides a solid ground to study the relations between 4D and 5D Dirac-like equations, since the space of complex spinors in $\clunotres$ is isomorphic to the space of real spinors in $\cldostres$.

An alternative possibility for complexification has been studied in \cite{paravectors}, based on the isomorphism $\clunotresc\cong\cl_{4,1}(\R)$. In this approach, an additional spacelike dimension is used, but the signature of space and time coordinates has to be interchanged.

The complexification of the STA based on $\cldostres$ lead us to think about the nature of complex numbers in the Dirac theory, which in this case is linked directly to the existence of an extra timelike dimension. Indeed, the center of $\cldostres$ is identified with $\mathbb{C}$, while the center of the real STA is simply $\R$. 

From a physical point of view, extra dimensions have been widely used in the context of unification theories, cosmology and quantum gravity. In particular five dimensional models constitute the minimal extension of spacetime and are considered the low-energy limit of higher dimensional theories (e.g. 11D supergravity)\cite{sugra,sugra5d,wesson}.

A particular use of a 5D spacetime was made by Wesson and collaborators for the \emph{induced matter theory} or IMT\cite{wesson,wessonmash}. This theory was first conceived as a way to fully geometrize the Einstein's equations, as it models the matter content of spacetime as a geometric property of a vacuum (Ricci-flat) 5D manifold\cite{wesson}, but was later extended to the description of other physical theories--- for the Dirac equation it was studied how the mass of the spinor field can be induced from an extra dimension\cite{wesson}. In this regard one can consider that a Dirac spinor satisfies a massless Dirac equation in 5D and that the mass of the spinor is an eigenvalue of the momentum of the particle in the extra direction. Under this mechanism a massive 4D Dirac equation is obtained. The case has been studied for an extra spacelike dimension in different gravitational contexts\cite{ma,sanchez,mona,koci}, and also for an extra timelike dimension\cite{mona,arcodiaparantipar}.

Naturally, when working with multiple time dimensions the problems of tachyons and ghost fields have to be adressed. We have commented on this issues in the context of our formulation in Ref. \cite{arcodiaparantipar}. Physical theories with more than one timelike dimensions have been proposed in different contexts\cite{bars,wesson2t}, and also a considerable amount of work has been done on spaces of signature $(+---+)$ in the framework of anti de Sitter spaces \cite{ads1,ads2}. 

In the recent article\cite{arcodiaparantipar} we combined the induced matter aproach of 5D spinors with the idea that the complex structure of spacetime comes from the existence of the extra timelike dimension, and showed that pure particle and antiparticle rest-solutions of the 4D Dirac equation are eigenspinors of the generators of rotations in the plane of two times.

In the present article we extend our research by analyzing the properties of the Clifford algebra $\cldostresr$, the space of algebraic spinors in this algebra and the inner products that can be defined on spinors.

An important result developed in this article is an expression of parity and coordinate time reversal transformations that could help falsify the extra dimension hypothesis. We find that 4D parity and time reversal are only realizable by means of an adjoint action in $\cldostresr$ by elements that also transform the extra dimension in a non-trivial way. In the context of IMT this implies that the mass would change accordingly, providing in principle an observable difference with respect to time reversal and parity transformations of $\clunotresc$.


In section \ref{secelem} we review the basics of Clifford algebras, paying special atention to the Pin and Spin groups. This makes the article as self contained as possible.

In section \ref{seccl23} we introduce the real Clifford algebra $\cldostresr$, and in section \ref{seccomplex} we define the spacetime algebra and establish the link between this algebra and $\cldostresr$. In this section we also construct the embedding $\pinunotres\hookrightarrow\spindostres$, and analize the representations of parity and time reversal under this mapping.

In section \ref{secmatrep} we review the theory of matrix representations of Clifford algebras, algebraic spinors, and inner products on spinor spaces. In section \ref{secembed}, we apply this formalism to the Clifford algebra $\cldostresr$, reproducing the usual complex structure of $\clunotresc$ in terms of the real algebra $\cldostresr$.

Finally, in section \ref{secfinal} we draw some conclusions and prospects for the topic.

\section{Elementaries of Cifford algebras: the multivector structure}\label{secelem}

Given a real n-dimensional vector space $V$ with a bilinear symmetric form $g:{V}\times{V}\rightarrow{\mathbb{R}}$, we say that $\Phi:{V}\rightarrow{\mathbb{R}}$, defined by $\Phi(v)=g(v,v)$, is the associated \emph{quadratic form} and we call the pair $(V,\Phi)$ a quadratic space. If the form $g$ is non-degenerate we say the quadratic space is regular. Since $g(a,b)=\frac{\Phi(a+b)-\Phi(a)-\Phi(b)}{2}$, no information is lost when passing from $g$ to $\Phi$.

For any quadratic space, one can build an associative unitary real algebra $\cl(\Phi)$, called the \emph{Clifford algebra (CA) for $(V,\Phi)$}. It is possible to define that algebra in different equivalent ways; we shall do it as follows: Let $\{e_{1},..,e_{n}\}$ be a basis for the vector space $V$, and $g_{ij}$ the matrix elements of the bilinear form $g$ in the given basis. The CA is defined by the generators $\{E_{1},...,E_{n}\}$, with the relations:
\begin{equation}
E_{i}E_{j}+E_{j}E_{i}=2g_{ij}\mathbf{1},
\end{equation}
where $\mathbf{1}$ is the unit in $\cl(\Phi)$. We shall just mention that this algebra can also be constructed as a quotient algebra of the tensor algebra modulo certain ideal\cite{Gallier}, but we shall not go further in this subject. Since the tensor algebra of a real vector space is a real algebra, the CA obtained also happens to be a real algebra, although it can be complexified. In this article both real and complex Clifford algebras are going to be relevant.

For any non-degenerate bilinear symmetric form $g$ we can find an orthonormal basis such that it has $p$ vectors with norm $1$ and $q$ vectors with norm $-1$. We call the pair $(p,q)$ \emph{the signature} of $g$. Sometimes, when working in a particular basis, we represent the signature as a $n$-tuple of plus and minus signs; for instance if we are in a space of signature $(2,3)$ in an orthonormal basis in which the first and last vectors have positive norm and the rest have negative norm, we write the signature as $(+,-,-,-,+)$. Because the signature is a geometric invariant and the CA is independent of the choice of basis, we use the notation $\clpqr$ for the real Clifford algebra of signature $(p,q)$, $\clpqc$ for the complex one, and $\clpq$ to express that any of the two cases can be considered.

Because there is an injective function from $V$ to the CA, via $e_{i}\mapsto{E}_{i}$, by abuse of notation we shall refer to the generators of this algebra as $e_{i}$. In the same way, we shall refer to the subspace $\text{span}_{\mathbb{R}}\{E_{1},...,E_{n}\}\subseteq\cl(\Phi)$, as $V$. It happens that this algebra is finite dimensional with dimension $2^{n}$ \cite{Lounesto}, and provided a basis $\{e_{1},...,e_{n}\}$ for $V$, the set of monomials:
\begin{equation}
e_{i_{1}}\dots{e}_{i_{k}} \ : \ 1\leq{i_{1}}<i_{2}<...<i_{k}\leq{n},
\end{equation}
where $k$ runs from $1$ to $n$, together with the identity of the Clifford algebra, $\mathbf{1}$, form a basis for the Clifford algebra, called the \emph{standard basis} associated to the basis $\{e_{1},...,e_{n}\}$ of V.

Since the Clifford algebra has the same dimension that the exterior algebra, we can conclude that they are isomorphic as vector spaces. Furthermore, it is possible to define a wedge product in the Clifford algebra in order to replicate the structure of $k$-vectors of the exterior algebra. We will say that a scalar (i.e. an element proportional to $\mathbf{1}$) is a $0$-vector; that a vector is a $1$-vector; and for any $k\in\{2,...,n\}$ we will define the space of \emph{$k$-vectors}, ${\bigwedge}^{k}$ as:
\begin{equation}
{\bigwedge}^{k}V=\{A^{i_{1}...i_{k}}u_{i_{1}}\dots{u}_{i_{k}} \ | \ g(u_{l},u_{m})=0, \ \ \forall{l,m\in\{i_{1},...,i_{k}\}}\},
\end{equation}
where the sum is finite, and each $u_{l}$ is a vector for $l\in\{i_{1},...,i_{k}\}$. It can be seen that these spaces are linear, and have dimension ${n}\choose{k}$. The set of $2$-vectors (or \emph{bivectors}) is important since it spans the Lie algebra of the $\operatorname{Spin}(\Phi)$ group. The $n$-vectors are also called \emph{pseudoscalars}, and the $(n-1)$-vectors, \emph{pseudovectors}.

An element which is the sum of different types of $k$-vectors (for example the sum of a 1-vector and a bivector) is called a \emph{multivector}, and in consequence $k$-vectors are called \emph{homogeneous multivectors of degree k}. Indeed we have that any element of the Clifford algebra is a multivector, namely:
\begin{equation}
\cl(\Phi)=\bigoplus_{k=0}^{n}{\bigwedge}^{k}V\cong{\bigwedge}V,
\end{equation}
where $\bigwedge{V}$ is the exterior algebra of $V$ and it is isomorphic to the Clifford algebra only in a vector space sense. Furthermore the direct sum appearing in this equation is also a direct sum of vector spaces, since in general spaces of homogeneous multivectors don't posses an algebra structure.

This decomposition allows us to write any element $A\in{\cl}(\Phi)$ as:
\begin{equation}\label{multiv}
A=\sum_{k=0}^{n}\langle{A}\rangle_{k},
\end{equation}
where the $k$-vectors $\langle{A}\rangle_{k}\in{\bigwedge^{k}}$ are uniquely determined. 
\begin{obs}
The structure of $k$-vectors of the Clifford algebra is more easily observed once an orthogonal basis is chosen. In this case, the structure is straightforwardly seen in the standard basis, as illustrated in Figure \ref{general}.
\begin{figure}[h]
\includegraphics[width=0.5\textwidth]{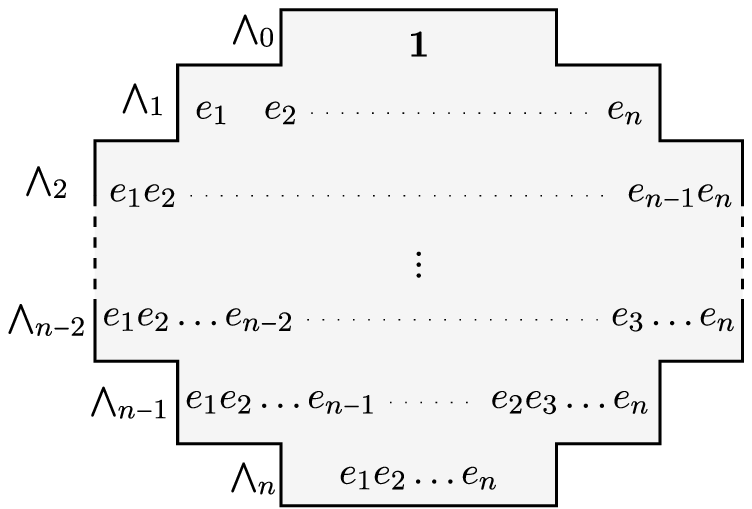}
\caption{$\mathbb{Z}_{n}$-grading of the CA, reproduced by an orthogonal basis.}\label{general}
\end{figure}\\
If we had a more general basis, $\{v_{1},...,v_{n}\}$, we can achieve a similar diagram by considering the basis $\{\mathbf{1}\}\cup\{v_{i_{1}}\wedge...\wedge{v_{i_{k}}} \ : \ 1\leq{i_{1}}<i_{2}<...<i_{k}\leq{n}\}$, which is different from the standard basis associated to $\{v_{1},...,v_{n}\}$. Due to the fact that ${{n}\choose{n-k}}={{n}\choose{k}}$, this diagram will always form a shape called \emph{guemil}\footnote{N.B.: Pronounced /gemil/ in international phonetic alphabet, which is close to the English pronunciation of the word ``gaemil''. It is a symbol present in the culture of the Mapuche people, which is a native people from the region of Patagonia, in the south of Argentina and Chile. As I searched for a name for this shape this was the one I found more accurate.}.
\end{obs}

We define the \emph{grade involution} $\alpha$ on a basis element  $e_{i_{1}}\dots{e}_{i_{k}}$ as: $\alpha({e_{i_{1}}\dots{e}_{i_{k}}})=(-1)^{k}e_{i_{1}}\dots{e}_{i_{k}}$, and extend it to any element as an algebra homomorphism on $\cl(\Phi)$.

It can be proven that this morphism induces a $\mathbb{Z}_{2}$ grading in the algebra, splitting it into $\cl(\Phi)=\cl(\Phi)^{0}\oplus\cl(\Phi)^{1}$, where $\cl(\Phi)^{i}=\{x\in\cl(\Phi):\alpha(x)=(-1)^{i}x\}$. It is true that $\cl(\Phi)^{0}$ is a subalgebra of $\cl(\Phi)$ while $\cl(\Phi)^{1}$ is not. Furthermore, it happens that
\begin{equation}\nonumber
\cl(\Phi)^{0}=\bigoplus_{i \ \text{even}}{\bigwedge}^{i}V \quad \text{and} \quad \cl(\Phi)^{1}=\bigoplus_{i \ \text{odd}}{\bigwedge}^{i}V .
\end{equation}
We define \emph{reversion}, $t$, on a basis element $e_{i_{1}}\dots{e}_{1_{k}}$ as $t(e_{i_{1}}\dots{e}_{i_{k}})=e_{i_{k}}\dots{e}_{i_{1}}$, and extend it as an algebra anti-morphism (meaning $t(ab)=t(b)t(a)$, details can be found in \cite{Gallier, VdR}). Reversion on an element $x$ will also be noted as $\hat{x}$.

Using the previous functions we define the \emph{Clifford conjugation} on any CA element $x$ as $\overline{x}=(t\circ\alpha)(x)=(\alpha\circ{t})(x)$. This is an algebra anti-morphism and in a basis element $e_{i_{1}}\dots{e}_{i_{k}}$ it can be seen to be $\overline{e_{i_{1}}\dots{e}_{i_{k}}}=(-1)^{k}e_{i_{k}}\dots{e}_{i_{1}}$.

Lastly, we define the \emph{norm} $N(x)$ of an element $x$ in the CA as $N(x):=x\overline{x}=\overline{x}{x}$. An important feature of this function is that $N(v)=-\Phi(v)\mathbf{1}$ for any $v\in{V}$.

\subsection{The Clifford-Lipschitz and the twisted Clifford-Lipschitz groups}\label{CLgroup}

The group of units of the Clifford algebra, $\cl(\Phi)^{*}$, contains some special subgroups that are related to the isometries of the quadratic space $(V,\Phi)$. In this section we will briefly define these groups and introduce some concepts related to them. We are not going to prove most of the results stated in this section, but detailed calculations can be found in \cite{VdR}.

The \emph{Clifford-Lipschitz group} or simply the \emph{Clifford group}$, \Gamma(\Phi)$, is defined as follows:
\begin{equation}
\Gamma(\Phi)=\{x\in\cl(\Phi)^{*}: xvx^{-1}\in{V}, \ \forall \ v\in{V}\}.
\end{equation}

As it is explained in \cite{VdR} the Lie algebra of this Lie group is the set:
\begin{equation}
\gamma(\Phi)={\bigwedge}^{2}V\oplus{\mathcal{Z}(\cl(\Phi))},
\end{equation}
where $\mathcal{Z}(\cl(\Phi))$ is the center of the CA and the bracket of the Lie algebra bracket is the commutator $[x,y]:=xy-yx$.We will denote the Clifford-Lipschitz group of the bilinear form of signature $(p,q)$ by $\Gamma_{p,q}$.

There is a natural action of the Clifford group on the vector space $V=\bigwedge_{1}$, $\opn{Ad}:\Gamma(\Phi)\rightarrow\opn{Aut}(V)$, called the \emph{adjoint action}, given by:
\begin{equation}
\opn{Ad}_{x}(v)=xvx^{-1}.
\end{equation}

It can be seen that $\opn{Ad}(\Gamma(\Phi))=\opn{O}(\Phi)$ if $p+q$ is even, and $\opn{Ad}(\Gamma(\Phi))=\opn{SO}(\Phi)$ if $p+q$ is odd. In any case we have that $\opn{ker(Ad)}=\mathcal{Z}^{*}(\cl(\Phi)):=\mathcal{Z}(\cl(\Phi))\cap\cl^{*}(\Phi)$.

In order to generate the whole orthogonal group for $p+q$ odd, it is possible to define the \emph{twisted adjoint action}, and the \emph{twisted Clifford-Lipschitz group} or \emph{twisted Clifford group}. The twisted Clifford-Lipschitz group is given by:
\begin{equation}
\hat\Gamma(\Phi)=\{x\in\cl(\Phi)^{*}: \alpha(x)vx^{-1}\in{V}, \ \forall \ v\in{V}\}.
\end{equation}

The \emph{twisted adjoint action} is defined as the function $\hat{\opn{Ad}}:\hat{\Gamma}(\Phi)\rightarrow{\opn{Aut}(V)}$, where for any $x\in{\hat\Gamma}(\Phi)$ and $v\in{V}$ we have:
\begin{equation}
\hat{\opn{Ad}}_{x}(v)=\alpha(x)vx^{-1}.
\end{equation}

For any Clifford algebra it is true that $\hat{\opn{Ad}}(\hat{\Gamma}(\Phi))=\opn{O}(\Phi)$ and that $\ker(\hat{\opn{Ad}})=\mathbb{R}^{*}\mathbf{1}$.

In general we have that $\hat{\Gamma}(\Phi)\subseteq{\Gamma(\Phi)}$, and that $\hat{\Gamma}(\Phi)=\Gamma(\Phi)$ only if the dimension of the vector space $V$ is even\cite{crumey}.

Since the twisted adjoint action seem to work better regarding the orthogonal groups, we may think that it would be possible to work only with this action and forget about $\opn{Ad}$. However, when the Clifford algebra is interpreted as the algebra of linear transformations on the space of spinors, it is only the adjoint action the one that reproduces the change of a matrix $x\in{\cl(\Phi)}$ when the spinors are transformed by a linear transformation.

\subsection{The Pin and Spin groups}

The mapping $\hat{\opn{Ad}}$ defined over the twisted Clifford-Lipschitz group can be restricted to a certain subgroup in a way in which this function remains surjective. 

The \emph{Pin group} of the CA $\cl(\Phi)$ is defined as the set:
\begin{equation}
\text{Pin}(\Phi)=\{x\in{\hat{\Gamma}(\Phi) \ |  \ N(x)=\pm\mathbf{1}}\}.
\end{equation}
The \emph{Spin group} of the CA $\cl(\Phi)$ is the subgroup of the Pin group, defined as
\begin{equation}
\text{Spin}(\Phi)=\opn{Pin}(\Phi)\cap\cl^{0}(\Phi).
\end{equation}
We will refer to the Pin group of signature $(p,q)$ by $\text{Pin}(p,q)$, and to the corresponding Spin group by $\opn{Spin}(p,q)$.

\begin{theo}\label{teotwAd}
The following statements are true:
\begin{itemize}
\item The restriction of the twisted adjoint action to the Spin group \linebreak{$\hat{\opn{Ad}}:\opn{Spin}(\Phi)\rightarrow\opn{SO}(\Phi)$} is a surjective group morphism with $\ker(\hat{\opn{Ad}})=\{-1,+1\}=\mathbb{Z}_{2}$.
\item The restriction of $\hat{\opn{Ad}}$ to $\opn{Pin}(\Phi)$ is a surjective morphism \linebreak{$\hat{\opn{Ad}}:\opn{Pin}(\Phi)\rightarrow\opn{O}(\Phi)$} with $\ker(\hat{\opn{Ad}})=\mathbb{Z}_{2}$.
\end{itemize}
It is said that $\opn{Spin}(\Phi)$ is the double cover of $\opn{SO}(\Phi)$, while $\opn{Pin}(\Phi)$ is the double cover of $\opn{O}(\Phi)$.
\end{theo}

Let us note that in the definition of the Spin group the twisted Clifford group can be replaced by the Clifford group and the twisted adjoint action by the adjoint action. For the Pin group, this change in the definition is only possible in even dimensions, where $\Gamma(\Phi)=\hat{\Gamma}(\Phi)$.

For odd-dimensional spaces we have that the Clifford group only covers $\opn{SO}(p,q)$, and so does $\opn{Pin}(p,q)$ (with the adjoint action). In this case, the Pin group defined as above, together with the adjoint action, wouldn't add any extra information to the Spin group: it only enlarges the kernel of the adjoint action. We have the following theorem:

\begin{theo}\label{teoAd}
The following statements hold:
\begin{itemize}
\item The restriction of the adjoint action to the Spin group \linebreak{${\opn{Ad}}:\opn{Spin}(\Phi)\rightarrow\opn{SO}(\Phi)$} is a surjective group morphism with $\ker({\opn{Ad}})=\mathbb{Z}_{2}$.
\item For even dimensional spaces, the restriction of ${\opn{Ad}}$ to $\opn{Pin}(\Phi)$ is a surjective morphism \linebreak{${\opn{Ad}}:\opn{Pin}(\Phi)\rightarrow\opn{O}(\Phi)$} with $\ker(\opn{Ad})=\mathbb{Z}_{2}$.
\item For odd dimensional spaces, the restriction of ${\opn{Ad}}$ to $\opn{Pin}(\Phi)$ is a surjective morphism \linebreak{${\opn{Ad}}:\opn{Pin}(\Phi)\rightarrow\opn{SO}(\Phi)$} with 
\begin{equation}
\ker(\opn{Ad})=
\begin{cases}
\mathbb{Z}_{2}\times\mathbb{Z}_{2}\ \ \ &\text{if }\ p-q\equiv1\mod{4},\\ 
\mathbb{Z}_{4} \ \ \ &\text{if }\ p-q\equiv3\mod{4}. \end{cases}
\end{equation}
\end{itemize}
\end{theo}

\begin{proof}
$\opn{Ad}$ coincides with $\hat{\opn{Ad}}$ on $\opn{Spin}(\Phi)$, hence the first statement is equivalent to Thm. \ref{teotwAd}. For the second and third statements of the theorem, we know that $\ker{\opn{Ad}|_{\Gamma(\Phi)}}=\mathcal{Z}^{*}(\cl(\Phi))$, hence $\ker{\opn{Ad}|_{\opn{Pin}(\Phi)}}$ is the set of elements that lie in $\mathcal{Z}^{*}(\cl(\Phi))\cap\opn{Pin}(\Phi)$.

Any element in set $\mathcal{Z}^{*}(\cl(\Phi))$ can be written as $z=a+Ib$ with $a,b\in{\mathbb{R}}$, and $I$ the pseudoscalar of the algebra. We have that $\overline{z}=a-bI$ or $\overline{z}=a+Ib$, according to $\overline{I}=\pm{I}$. However, since $\alpha(x)=\pm{x}$ for every element in $\opn{Pin}(\Phi)$\footnote{This can be seen easily because any element $x$ in the Pin group is the product of a finite number of non-isotropic vectors. For a product of an even number of vectors $\alpha(x)=x$ and for an odd number of vectors $\alpha(x)=-x$.}, $z$ is rather equal to $a$ or equal to $b{I}$, and hence $z\overline{z}=a^{2}$ or $z\overline{z}=\pm{I}^{2}b^{2}$. Consequently, the solutions to $z\overline{z}=\pm{1}$ are always $z=\pm{1}$ and $z=\pm{I}$, this is:
\begin{equation}
\ker{\opn{Ad}|_{\opn{Pin}}}=\{\pm{\mathbf{1}},\pm{I}\}.
\end{equation} 
If $I^{2}=1$, equivalently $p-q\equiv{1}\mod4$, this set has the group structure of $\mathbb{Z}_{2}\times\mathbb{Z}_{2}$. If $I^{2}=-1$, which is equivalent to $p-q\equiv{3}\mod4$, the same set has the group structure of $\mathbb{Z}_{4}$.
\end{proof}

We are going to define the reduced Pin group\cite{VdR} as:
\begin{equation}
\tensor{\opn{Pin}}{^{\wedge}_{+}}(\Phi)=\{x\in\hat{\Gamma}(\Phi) \ | \ x\bar{x}=\mathbf{1}\}.
\end{equation}
It can be proved that this group (using the twisted adjoint action) generates the transformations that preserve the orientation of the subspace $\mathbb{R}^{p,0}$: for instance, in $\mathbb{R}^{1,3}$ these are the transformations preserving time orientation, $\opn{O}^{\uparrow}(1,3)$.

\section{The real Clifford algebra $\cl_{2,3}(\mathbb{R})$}\label{seccl23}

We denote by $\R^{2,3}$ the five dimensional vector space $\mathbb{R}^{5}$ together with the metric of signature $(+,-,-,-,+)$. A point in this space is labeled by coordinates $(x^{0},x^{1},x^{2},x^{3},x^{4})$. To represent any of the coordinates we use uppercase Latin scripts $x^{A}$, and hence $A$ ranges from $0$ to $4$. The matrix of the bilinear form $g$ is given in this orthonormal basis by:
\begin{equation}
g_{AB}=
\begin{cases}
1 & \text{if $A=B=0$ or $A=B=4$},\\
-1 & \text{if $A=B\in\{1,2,3\}$},\\
0 & \text{otherwise}.
\end{cases}
\end{equation}
The structure of $k$-vectors is depicted by means of the standard basis in Figure \ref{cl23}.
\begin{figure}[h]
\includegraphics[width=\textwidth]{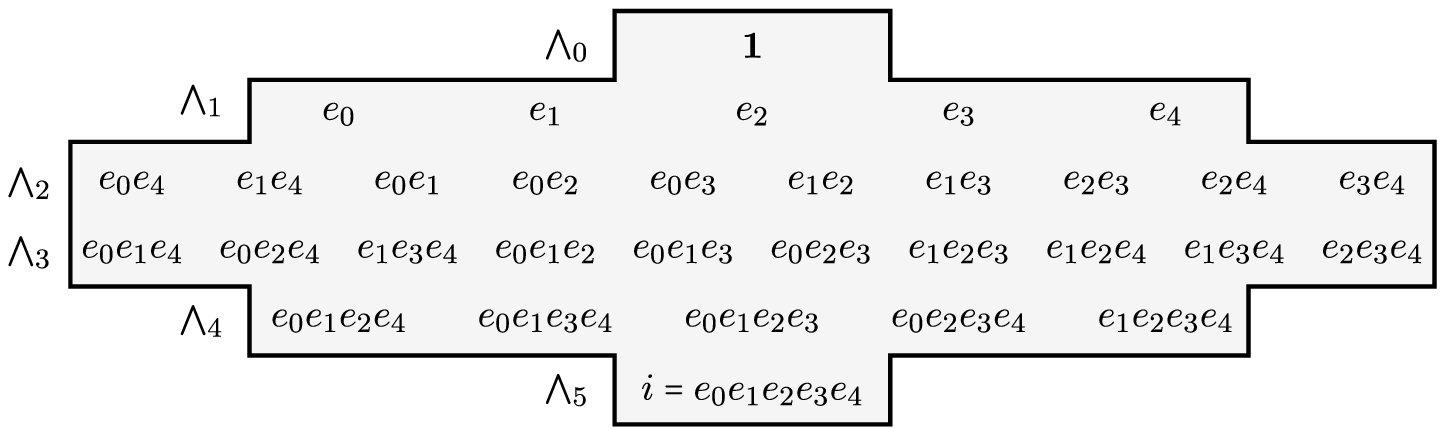}
\caption{Structure of $k$-vectors for $\cldostres$}\label{cl23}
\end{figure}

According to the Clifford algebra classification, this algebra is isomorphic to the algebra of $4\times{4}$ complex matrices, $\mathcal{M}(4,\mathbb{C})$. 

Observe that we have denoted the pseudoscalar of $\cldostresr$ by the symbol $i$. We find motivation for this notation in the fact that the pseudoscalar is a central element that squares to $-\mathbf{1}$, thus it is possible to identify it with an imaginary unit. Indeed, the center of this algebra is $\mathcal{Z}(\cl_{2,3}(\mathbb{R}))={\bigwedge}_{0}\oplus{\bigwedge}_{5}=\mathbb{R}\mathbf{1}\oplus\mathbb{R}i$ is isomorphic to ${\mathbb{C}}$. This is an important feature that allows us to reproduce complex numbers in the real algebra. From now on, the symbol $i$ will be reserved for the pseudoscalar of $\cldostresr$.

\section{The complex Clifford algebra $\cl_{1,3}(\mathbb{C})$ and the isomorphism $\cl_{2,3}(\mathbb{R})\cong\cl_{1,3}(\mathbb{C})$}\label{seccomplex}

In the theory of special relativity, spacetime is modeled as \emph{the Minkowski spacetime} $\R^{1,3}$. This is, the vector space $\mathbb{R}^{4}$, with the metric of signature $(+,-,-,-)$. A point is labeled by coordinates $(x^{0},x^{1},x^{2},x^{3})$, and we will use lowercase Greek scripts to refer to them arbitrarily: $x^\mu$ with $\mu$ taking values in $\{0,1,2,3\}$. Additionally, lowercase Latin scripts (e.g. $i,j,k$), will take values in $\{1,2,3\}$, thus $x^i$ will represent any of the space coordinates. The bilinear form in an orthogonal basis is written as: 
\begin{equation}
\eta_{\mu\nu}=
  \begin{cases}
    1 &  \ \mu=\nu=0, \\
   -1 & \ \mu=\nu=j ; \ j \in \{1,2,3\},\\
    0 & \mu\neq\nu.
  \end{cases}
\end{equation}
The real CA corresponding to this spacetime, $\cl_{1,3}(\mathbb{R})$, will be the algebra generated by elements $\{\gamma_{\mu}:\mu\in\{0,...,3\}\}$ with the relation:
\begin{equation}\label{CArelation}
\gamma_{\mu}\gamma_{\nu}+\gamma_{\nu}\gamma_{\mu}=2\eta_{\mu\nu}\mathbf{1}.
\end{equation}

According to the well known classification of CA, the real algebra for the Minkowski spacetime with this signature is isomorphic to the algebra of ${2}\times{2}$ matrices with entries in the quaternions.

The Dirac's theory of spinors is derived partially from classical quantum mechanics, which postulates the existence of a complex Hilbert space of physical states, hence we need to use the complex CA, $\cl_{1,3}(\mathbb{C})$. This algebra is well known to be isomorphic to the complex algebra of ${4}\times{4}$ matrices with complex entries, $\mathcal{M}(4,\mathbb{C})$. There are infinite matrix representations for the generators as matrices in $\mathcal{M}(4,\mathbb{C})$, but the most popular are perhaps the Dirac and Weyl representations.

In this article we will pay special attention to the following known fact: the CA of $\mathbb{R}^{1,3}$ can be ``complexified'' in an alternative way \cite{Lounesto}, which allow us to keep working with real Clifford algebras. The complexification is accomplished by adding an extra time-like dimension, $x^{4}$, to the Minkowski spacetime and taking the real CA of the 5D spacetime, $\mathbb{R}^{2,3}$, defined in section \ref{seccl23}. This is possible because the following isomorphisms hold:
\begin{equation}\label{iso}
\cl_{2,3}(\mathbb{R})\cong\mathcal{M}(4,\mathbb{C})\cong\cl_{1,3}(\mathbb{C}).
\end{equation}
The ismorphism $\clunotresc\cong\cldostresr$ is not unique; however, for any possible isomorphism the imaginary unit of $\clunotresc$ must be identified rather with the pseudoscalar $i=e_{0}e_{1}e_{2}e_{3}e_{4}$ of $\cldostresr$ or with its opposite, $-i$. This is possible because of two facts: on the one hand the pseudoscalar squares to $-\mathbf{1}$ and on the other hand it lies in the center of the algebra, $\mathcal{Z}\text{(Cl}_{2,3})$, this is, it commutes with every element in the CA (a necessary property for scalars). Indeed we have that $\mathcal{Z}(\cl_{2,3})\cong\mathbb{C}$.

It is important to be careful since, when working in $\cl_{2,3}(\mathbb{R})$, we shall use the name $i$, as referring to the pseudoscalar, but we are by no means complexifying $\cl_{2,3}(\mathbb{R})$. Reciprocally, when we complexify the theory in the usual way, this fifth dimension emerges naturally as the matrix element $\gamma_{5}$, which squares to $1$ and is associated with the chirality of the Dirac spinor fields. Since the isomorphism \eqref{iso} holds, the representations of $\cl_{1,3}(\mathbb{C})$ and $\cl_{2,3}(\mathbb{R})$ are equivalent, and hence the spaces of spinors are isomorphic.

In a recent article \cite{arcodiaparantipar}, we explored the physical consequences of considering this dimension as physically real, and interpreted the Dirac particles and antiparticles in this context. In this article we will further explore the mathematical structure of the spinor spaces, the group of automorphisms, and the embeddings of the $\opn{Pin}(1,3)$ and $\opn{Spin}(1,3)$ groups into $\opn{Spin}(2,3)$.

We turn now to the explicit construction of the isomorphism $\clunotresc\tilde{\rightarrow}\cldostresr$, and in order to do so we are going to consider a particular class of isomorphisms. First, note that any embedding of the vector space $\R^{1,3}$ into $\R^{2,3}$ induces an embedding of the real Clifford algebras $\clunotresr\hookrightarrow\cldostresr$. If we additionally identify the imaginary unit $I$ of the complex algebra $\clunotresc$ with rather the pseudoscalar $i=e_{0}e_{1}e_{2}e_{3}e_{4}$ or its opposite $-i$, then a full isomorphism $\clunotresc\tilde{\rightarrow}\cldostresr$ is determined.

We are going to consider two embeddings of this kind, namely the trivial embedding and the twisted embedding. \emph{The trivial embedding} consists on simply identifying the first four coordinates of $\R^{2,3}$ with the coordinates of $\R^{1,3}$, and the imaginary unit with the pseudoscalar of $\cldostresr$. This is:
\begin{equation}
\begin{gathered}
\gamma_{\mu}\mapsto{e}_{\mu},\\
I\mapsto{i}.
\end{gathered}
\end{equation}
By abuse of notation we will denote the image of $\R^{1,3}$ under this map also by $\R^{1,3}$. Furthermore, since $\bigwedge_{k}$ is mapped injectively to $\bigwedge_{k}\subseteq\cldostresr$, we will refer to its image also by $\bigwedge_{k}$, as long as there's no confusion.

\emph{The twisted embedding} consists on identifying the vectors in $\R^{1,3}$ with some particular $3$-vectors in $\cldostresr$. This embedding appeared in our work \cite{arcodiaparantipar}, when the massive $4D$ Dirac equation was obtained from a massless Dirac equation in $5D$. The isomorphism obtained in this way consists on:
\begin{equation}
\begin{split}
\gamma_{\mu}&\mapsto\tilde{e}_{\mu}:=-ie_{4}e_{\mu},\\
I&\mapsto{i}.
\end{split}
\end{equation}
more specifically we have:
\begin{equation}
\tilde{e}_{0}=e_{1}e_{2}e_{3}\ ; \ \tilde{e}_{1}=e_{0}e_{2}e_{3}\ ; \ \tilde{e}_{2}=e_{1}e_{0}e_{3}\ ; \ \tilde{e}_{3}=e_{0}e_{1}e_{2} \ .
\end{equation}
We will refer to the image of $\R^{1,3}$  under this embedding by $\widetilde{\R^{1,3}}$, and more generally we will call $\tilde{\bigwedge_{k}}$ the image of $\bigwedge_{k}$, where $k\in\{0,1,2,3,4\}$. Note that since $\tilde{e}_{\mu}\tilde{e}_{\nu}={e}_{\mu}{e}_{\nu}$, then the set $\{\tilde{e}_{\mu}\}$ generates the same algebra that $\{{e}_{\mu}\}$. Note also that although $2k$-vectors are mapped to $2k-$vectors in 5D, vectors and $3-$vectors are swapped. All this facts are pictured in Figure \ref{embebimiento}. The shaded area in this diagram represent the image of $\clunotresr$ under any of the embeddings.

\begin{figure}[h!]
\includegraphics[width=\textwidth]{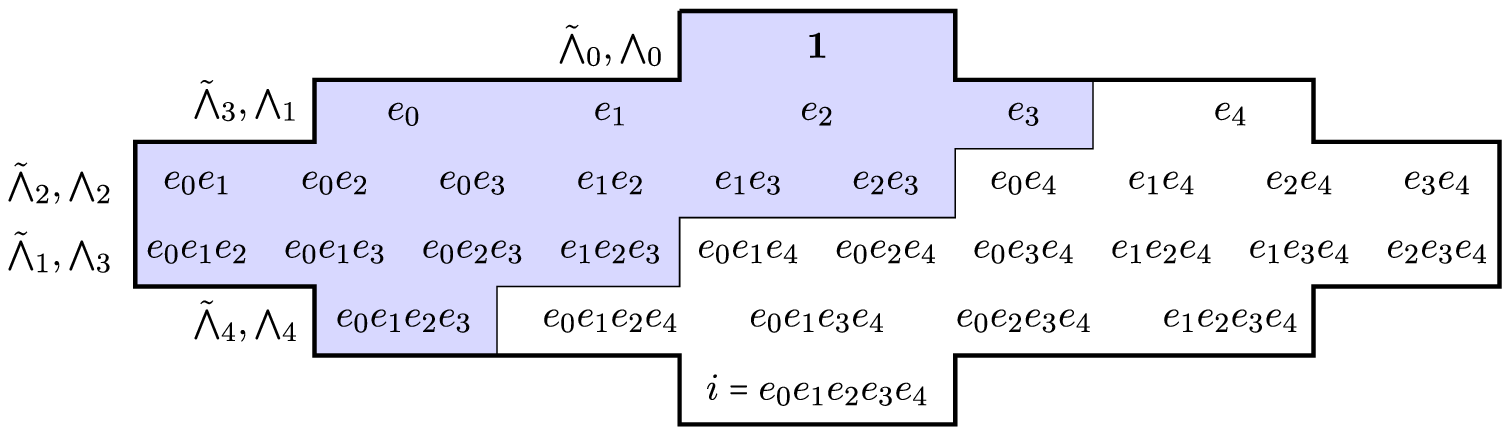}
\caption{The embedding $\cl_{1,3}(\mathbb{R})\hookrightarrow{\cl_{2,3}(\mathbb{R})}$}\label{embebimiento}
\end{figure}

\subsection{The group $\opn{Pin}(1,3)$ and its inclusion in $\opn{Spin}(2,3)$}

Since our base space has five dimensions, by means of Thm. \ref{teoAd}, the adjoint action over the Clifford-Lipschitz group only covers the special orthogonal group $\opn{SO}(2,3)$, and hence no orientation changing linear transformation can be achieved by means of the Clifford-Lipschitz group. However, it is possible to reproduce all the orthogonal transformations of the subspace $\mathbb{R}^{1,3}$ generated by the first four coordinates. Indeed, let's note that by means of the trivial embedding, the vector space $\mathbb{R}^{2,3}$ can be decomposed as:
\begin{equation}\label{inclusionV}
\mathbb{R}^{2,3}=\mathbb{R}^{1,3}\oplus\mathbb{R}e_{4},
\end{equation}
where the bilinear form $g$ in $\mathbb{R}^{2,3}$ can be defined for a pair of elements $v+\alpha{e}_{4}$ and $u+\beta{e}_{4}$, where $u,v\in{\mathbb{R}^{1,3}}$ as:
\begin{equation}
g(v+\alpha{e}_{4},u+\beta{e}_{4})=\eta{(v,u)}+\alpha\beta,
\end{equation}
with $\eta$ the bilinear form of signature $(1,3)$.

Let $L$ be an element in $\opn{O}(1,3)$, we are going to extend it to a linear transformation in $\mathbb{R}^{2,3}$, as follows:
\begin{equation}
\tilde{L}(v+\alpha{e}_{4})=Lv+\opn{det}{(L)}\alpha{e}_{4}.
\end{equation}
It can be checked by direct computation that $\tilde{L}\in{\opn{SO}(2,3)}$: this is $\tilde{L}$ preserves the bilinear form $g$ and, independently of the value of $\opn{det}{L}$, the determinant of $\tilde{L}$ is always $1$. In this sense, we have the inclusion:
\begin{equation}
\ell:\opn{O(1,3)}\hookrightarrow\opn{SO}(2,3).
\end{equation}
Furthermore, note that under the splitting given by Eq. \eqref{inclusionV} the group $\opn{Pin}(1,3)$ is contained in $\clunotres(\mathbb{R})\subseteq\cldostres(\mathbb{R})$, and if we consider the mapping $\theta:\opn{Pin}(1,3)\rightarrow\opn{Spin}(2,3)$ as follows:
\begin{equation}\label{tita}
\theta{(x)}:=\begin{cases}
		x \ \ \text{if $x\in\opn{Spin}(1,3)$},\\
		ix \ \ \text{if $x\notin\opn{Spin}(1,3)$},
		\end{cases}
\end{equation}
where $i$ is the pseudoscalar of $\cldostres$, we have that $\Ad_{\theta(x)}=\Ad_{x}$ for all $x$. Due to this facts, the following diagram is commutative:
\begin{equation}\label{diag}
\begin{tikzcd}[sep=5em]
\opn{Pin(1,3)} \arrow[twoheadrightarrow]{d}{\opn{Ad}} \arrow[hookrightarrow]{r}{\theta}
& \opn{Spin}(2,3) \arrow[twoheadrightarrow]{d}{\opn{Ad}}\\
\opn{O}(1,3) \arrow[hookrightarrow]{r}{\ell}
& \opn{SO}(2,3) 
\end{tikzcd}
\end{equation}

As we see, by adding an extra time-like dimension an element in $\opn{Pin}(1,3)$ that lies in the odd part of the Clifford algebra, can be represented by an element in $\opn{Spin}(2,3)$ which is a subset of the even subalgebra of $\cldostres(\mathbb{R})$. This fact, albeit impressive is not new: it is analogous to the known property that orientation reversing transformations in $\R^{2}$ can be produced by orientation preserving transformations in $\R^3$, when $\R^2$ is considered as a subspace of $\R^{3}$. 

So far we have considered the trivial inclusion of $\mathbb{R}^{1,3}$ into $\mathbb{R}^{2,3}$ by its first four coordinates. But, as discussed earlier, we also have the twisted embedding, whose image $\widetilde{\mathbb{R}^{1,3}}$ is spanned by $\{\tilde{e}_{\mu}\ : \ \mu\in\{0,...,3\}\}$. The isomorphism between $\R^{1,3}$ and $\widetilde{\R^{1,3}}$ is given by $e_{\mu}\mapsto{\tilde{e}_{\mu}}$. Through this isomorphism, we can see what the action of an element in $\opn{Pin}(1,3)$ is on the elements of $\widetilde{\mathbb{R}^{1,3}}$. It is clear that $\opn{Ad}_{x}(\tilde{e}_{\mu})=\widetilde{\opn{Ad}_{x}(e_{\mu})}$ if $x$ is an element of $\opn{Spin}(1,3)$. Because of this, the transformation of components is going to be the same for a vector in $\mathbb{R}^{1,3}$ than for an element in $\widetilde{\mathbb{R}^{1,3}}$. Let us see what is the adjoint action of an element in $\opn{Pin}(1,3)\setminus\opn{Spin}(1,3)$, on an element $\tilde{v}=v^{\mu}\tilde{e}_{\mu}$:
\begin{align}\nonumber
\opn{Ad}_{x}(v^{\mu}\tilde{e}_{\mu})&=v^{\mu}x\tilde{e}_{\mu}x^{-1}=-v^{\mu}xie_{4}e_{\mu}x^{-1}=-v^{\mu}ixe_{4}x^{-1}xe_{\mu}x^{-1}=&\\
&=v^{\mu}ie_{4}xe_{\mu}x^{-1}=-ie_{4}(-\opn{Ad}_{x}(v^{\mu}e_{\mu}))=-\widetilde{\opn{Ad}_{x}(v^{\mu}e_{\mu})}.&
\end{align}
Hence, if components of a vector $v^{\mu}e_{\mu}$ are transformed according to $v^{\mu}\mapsto{v'^{\mu}}$, the components of the tilded vector $\tilde{v}={v}^{\mu}\tilde{e}_{\mu}$ transform according to $v^{\mu}\mapsto{-v'^{\mu}}$. Note that in signature $(1,3)$ the Clifford group and the twisted Clifford groups are equal\cite{crumey}, then the results just exposed prove that the adjoint action of this group on $\widetilde{\mathbb{R}^{1,3}}$ corresponds to the twisted adjoint action of the same group on $\mathbb{R}^{1,3}$.

\subsubsection{Parity and time reversal}

Two important elements of the Lorentz group $\opn{O}(1,3)$ are the parity transformation $P$ and the time reversal transformation $T$.

Indeed, it is possible to consider two types of time reversal transformations\cite{dewitt,winger}: one being unitary and the other anti-unitary. The anti-unitary one is also referred to as \emph{motion reversal}. The unitary time reversal is contained in the Lorentz group and we could call it a \emph{coordinate time reversal}, since its effect is to change any $4$-vector $(A^{0},\mathbf{A})$ to $(-A^{0},\mathbf{A})$. The anti-unitary time reversal is more similar to a proper time reversal which inverts the order of events on a timelike curve. Here will only consider the \emph{coordinate time reversal}.

The elements in $\opn{Pin}(1,3)$ that are mapped to a parity transformation under the adjoint action are $\pm{\gamma}_{0}$, as it is customary, we pick $P=\gamma_{0}$\cite{ramond}. For time reversal we pick $T=\gamma_{1}\gamma_{2}\gamma_{3}$\cite{dewitt} as it is mapped to coordinate time reversal under $\opn{Ad}$.

Using the mapping in Eq. \eqref{tita} we can find elements $\mathcal{P}$ and $\mathcal{T}$ in $\opn{Spin}(2,3)$ that reproduce the action of $P$ and $T$ respectively when they are restricted to $\mathbb{R}^{1,3}$. We obtain:
\begin{align}
\mathcal{P}&=e_{1}e_{2}e_{3}e_{4},\\
\mathcal{T}&=-e_{0}e_{4}.
\end{align}

Let us note that this elements invert also the last component of a vector in $\mathbb{R}^{2,3}$, and also that $\mathcal{T}=\exp(-\frac{\pi}{2}e_{0}e_{4})$ is connected to $\mathbf{1}$, while $T$ is not. In the treatment of the Dirac equation provided in Ref. \cite{arcodiaparantipar}, we identified the mass of the field to be proportional to the fifth component of the 5-momentum vector, hence this would imply that mass is inverted under both, time reversal and parity, but would be conserved under $\mathcal{PT}$. A deeper analysis of this situation is beyond the scope of this article but will be study in the future.

\section{The matrix representation and algebraic spinors}\label{secmatrep}

In this section we will introduce definitions and well known results regarding algebraic spinors in general Clifford algebras, in order to turn our attention to the real Clifford algebra $\cl_{2,3}$ in the next section. Since the algebras $\clunotresc$, $\clunotresr$ and $\cldostresr$ are simple algebras, we are going to consider mostly this case. The contents of this section are thoroughly explained in chapter 6 of \cite{VdR} and chapters 17 and 18 of \cite{Lounesto}.

Every Clifford algebra is rather a simple algebra or the direct sum of two isomorphic simple algebras. According to the Artin-Wedderburn theorem, every finite dimensional simple algebra is a matrix algebra, $\mathcal{M}(k,\mathbb{K})$, over certain division ring (or skew field), $\mathbb{K}$. This set is also isomorphic to the algebra of $\mathbb{K}$-endomorphisms, $\opn{End}_{\K}(S)$, of a vector space $S$ isomorphic to $\mathbb{K}^{k}$. Hence, every CA is rather isomorphic to $\mathcal{M}(k,\mathbb{K})$ or to $\mathcal{M}(k,\mathbb{K})\oplus\mathcal{M}(k,\mathbb{K})$, with suitable choices of $\K$, $k$ and $S$.

It is possible to construct $\K$ and $S$, and to obtain $k$, using uniquely elements in the algebra. However, this construction is not unique.

\subsection{The division ring $\mathbb{K}$}

Let us explore the possible ways to intrinsically construct the division ring $\K$ using elements of the Clifford algebra alone. The fact that $\mathcal{M}(k,\K)$ is a $\K$-bimodule, and the need to reproduce the product by scalars on $\mathcal{M}(k,\K)$, make it possible to have two definitions of $\K$--- one using a \emph{complete set of orthogonal idempotents}, and the other using just a single primitive idempotent. The later may seem simpler, but it produces a complicated rule for the product by scalars on $\mathcal{M}(k,\K)$. In contrast, the first possibility results in a more complex definition, but the product by scalars coincides with the associative product of the algebra. We are going to explore both of these constructions and use one or the other depending on the context.

In order to build the division ring $\mathbb{K}$ in a way in which the product by scalar on $\mathcal{M}(k,\mathbb{K})$ can be obtained by means of the Clifford product, we need to consider a \emph{complete set of orthogonal primitive idempotents}, $\{f_{1},...,f_{k}\}$ in the CA. This is, each $f_{i}$ is a primitive idempotent, the equation $f_{i}f_{j}=\delta_{ij}f_{i}$ holds for all ${i,j}$; and all the idempotents add up to 1. Every Clifford algebra possesses such a set. Now we define $\mathbb{K}$ to be
\begin{equation}\label{divalgebradef}
\mathbb{K}:=\left\{\sum_{i=1}^{k}f_{i}af_{i}\ |\ a\in\cl_{p,q}\right\}.
\end{equation}
We see that since $1=f_{1}+...+f_{k}\in\mathbb{K}$, then the unit of the division ring coincides with the unit of the Clifford algebra and given $\lambda\in{\mathbb{K}}$ and $a\in\cl_{p,q}$ we have that $\lambda{a}$ behaves like a product by a scalar.

Alternatively, it is customary to define $\mathbb{K}$ by picking a single primitive idempotent $f$, and taking it to be the set $f\cl_{p,q}f$. Although this construction is correct since $f\cl_{p,q}f\cong\mathbb{K}$, the unit of this set is $f$, and not $1$. This complicates the process of defining the $\Bbb{K}$-module structure on $\clpq$ since $f$ doesn't act like $1$ for all elements in $\clpq$.
In order to properly define the product by scalars it is necessary to extend the idempotent $f$ to a complete set of orthogonal idempotents. We are going to recall some basic facts and to use some formal results, but a full treatment of this construction is given in chapter 4 of \cite{VdR}. Let's call $f_{1}:=f$ and let $\{f_{1},...,f_{k}\}$ be the complete set of primitive orthogonal idempotents. For simple Clifford algebras we have that the sets $\mathcal{A}_{ij}:=f_{i}\cl_{p,q}f_{j}$ are all not null and they comply with:
\begin{equation}\label{Aijsets}
\mathcal{A}_{ij}\mathcal{A}_{lk}=\delta_{jl}\mathcal{A}_{ik},
\end{equation}
with the corresponding Peirce decomposition: 
\begin{equation}\label{peirce}
\cl_{p,q}=\bigoplus_{i,j=1}^{k}\mathcal{A}_{ij}.
\end{equation}
There exist a set 
\begin{equation}\label{baseconE}
\mathcal{B}=\left\{\mathcal{E}_{ij}\in{\mathcal{A}}_{ij}\ | \ i,j\in\{1,...,k\}\right\},
\end{equation}
such that $\mathcal{E}_{ij}\mathcal{E}_{lk}=\delta_{jl}\mathcal{E}_{ik}$, and $\mathcal{E}_{ii}=f_{i}$. The set $\mathcal{B}$ is a $\mathbb{K}$-basis for the Clifford algebra. It is possible to use the set $\mathcal{B}$ to construct an isomorphism between $\K$ and $f\clpq{f}$:
\begin{center}
\begin{minipage}[t]{0.4\textwidth}
\begin{align} \nonumber
\varphi:\K&\rightarrow{f}\clpq{f},\\ \nonumber
z&\mapsto{z}f,
\end{align}
\end{minipage}
\hspace{5mm}
\begin{minipage}[t]{0.4\textwidth}
\begin{align} \nonumber
\varphi^{-1}:f\clpq{f}&\longrightarrow\K, \\ \nonumber
\lambda&\longmapsto\sum_{j=1}^{k}\mathcal{E}_{j1}\lambda{\mathcal{E}}_{1j}.
\end{align}
\end{minipage}
\end{center}
Consequently the product by scalars when $\K$ is defined as $f\clpq{f}$ is given by:
\begin{equation}
\lambda\bullet{a}=\varphi^{-1}(\lambda)a=\sum_{j=1}^{k}\mathcal{E}_{j1}\lambda{\mathcal{E}}_{1j}a, \ \forall \lambda\in{f}\clpq{f},\ \forall{a}\in\clpq.
\end{equation}
From a matrix point of view this ambiguity in the definition of $\mathbb{K}$ can be seen in the following isomorphism:
\begin{equation}
\begin{pmatrix}
		\mathbb{K} & 0 & ... &0\\
		0 & 0 & ... &0\\
		\vdots & \vdots & \ddots & \vdots\\
		0 & 0 & ... & 0 \\
\end{pmatrix}
\cong
\mathbb{K}\begin{pmatrix}
		1 & 0 & ... &0\\
		0 & 1 & ... &0\\
		\vdots & \vdots & \ddots & \vdots\\
		0 & 0 & ... & 1 \\
\end{pmatrix}=\mathbb{K}\mathbf{1}.
\end{equation} 
The left hand side of the equivalence represent the set $f\clpq{f}$, while the right hand side is obtained by using the complete set of idempotents. It is clear that only the later reproduces the multiplication of a matrix by a scalar.
The cardinal of the complete set of orthogonal primitive idempotents, $k$, is given by $k=q-r_{q-p}$, where $r_{j}$ is the $j$-th \emph{Radon-Hurwitz number}, given in the following table:
\begin{equation}
\begin{tabular}{l*{8}{c}r}
$j$ &\vline&\ $0$\ &\ $1$\ &\ $2$\ &\ $3$\ &\ $4$\ &\ $5$\ &\ $6$\ &\ $7$\ \\
\hline
$r_{j}$ &\vline& \ 0 & 1 & 2 & 2 & 3 & 3 & 3 & 3 
\end{tabular}\ ,
\end{equation}
together with the rule $r_{j+8}=r_{j}+4$. $\mathbb{K}$ is always $\mathbb{R}$, $\mathbb{C}$ or $\mathbb{H}$, and the matrix form of the Clifford algebra is given by the well-known classification\cite{VdR,porteous,Lounesto}.

\subsection{The center of the Clifford algebra}

If we are dealing with a simple Clifford algebra, it is possible to use a known result from the theory of matrix algebra over the division ring $\mathbb{K}$. In this case the center of $\mathbb{K}$ and the center of the Clifford algebra are related by:
\begin{equation}
\mathcal{Z}(\cl_{p,q})\cong\mathcal{Z}(\mathcal{M}(k,\mathbb{K}))\cong{\mathcal{Z}(\mathbb{K})}.
\end{equation}
Hence, as long as $\mathbb{K}$ is not isomorphic to the quaternions, we have the isomorphism:
\begin{equation}\label{centerKCl}
\mathcal{Z}(\cl_{p,q})\cong{\mathbb{K}}.
\end{equation}
This equality is really useful, since the center of the Clifford algebras is easier to compute than $\K$. If the dimension of the basis vector space $V$ is even, we have that the center of the CA is isomorphic to $\mathbb{R}$ for real CA and to $\mathbb{C}$ for complex ones. If the dimension is odd, and the real CA is simple, then its center is given by the set $\bigwedge_{0}\oplus{\bigwedge_{n}}$ which is isomorphic to $\mathbb{C}$ for simple algebras. 

Indeed, given the definition for $\mathbb{K}$ given by Eq. \eqref{divalgebradef}, we have $\mathbb{K}=\mathcal{Z}(\cl_{p,q})$, and that the mapping $z\mapsto{z}f$ where $z$ is an arbitrary element in $\mathcal{Z}(\clpq)$ is an isomorphism.

\subsection{Spinors and semispinors}\label{spi}

As we explained above, if the Clifford algebra is simple, $\cl_{p,q}=\mathcal{M}(k,\mathbb{K})\cong{\opn{End}_{\mathbb{K}}{S}}$, and if it is semi-simple, $\cl_{p,q}=\mathcal{M}(k,\mathbb{K})\oplus\mathcal{M}(k,\mathbb{K})\cong{\opn{End}_{\mathbb{K}}{S}}\oplus{\opn{End}_{\mathbb{K}}{S}}$. The $\mathbb{K}$-linear space (or free $\mathbb{K}$-module $S$) has dimension $k$ and can be obtained easily from the Clifford algebra. Take a primitive idempotent $f$ on $\cl_{p,q}$, then $S=\cl_{p,q}f\cong{\mathbb{K}^{k}}$. In this way, we have a \emph{$\mathbb{K}$-linear representation of the Clifford algebra $\cl_{p,q}$}
\begin{align}
\rho:\cl_{p,q}&\rightarrow{\opn{End}_{\mathbb{K}}}(S)\\
a&\mapsto\rho(a):\psi{f}\mapsto{a\psi{f}},
\end{align}
which is called the \emph{left irreducible regular representation} of the Clifford algebra.

We call an element $\psi=\psi{f}\in{S}$ in the left irreducible regular representation, an \emph{algebraic spinor} if $\cl_{p,q}$ is simple, and an \emph{algebraic semispinor} if $\cl_{p,q}$ is not simple. Accordingly, the minimal left ideal $S$ will be referred to as \emph{the spinor space} or \emph{the semispinor space}, respectively.

In what follows we will assume that $\cl_{p,q}$ is simple.

In order to analyze the structure of the spinor space, it will be necessary to introduce a $\mathbb{K}$-basis for $S$. As long as we are consistent, $S$ can be regarded rather as a left or a right $\mathbb{K}$-module. In this article we will choose the latter. The right structure has the advantage of making possible to work with $\mathbb{K}$ by using only the set $f\cl_{p,q}f$, without the necessity of extending $f$ to a complete set of orthogonal primitive idempotents. This is due to the fact that $\psi{f}=\psi$ for all $\psi\in{S}=\cl_{p,q}f$, and $\cl_{p,q}f{f}\cl_{p,q}f={\cl_{p,q}}f$ for any primitive idempotent $f$\footnote{Because $S$ is a minimal left ideal and $SS=\cl_{p,q}f{f}\cl_{p,q}f$ is a left ideal contained in $S$, then $SS=S$.}.

Let the set $\{\mathbf{u}_{i}:i\in\{1,...,k\}\}$ be an ordered $\mathbb{K}$-basis for $S$, then any $\psi=af\in{S}$ can be written as:
\begin{equation}
\psi=af=\mathbf{u}_{i}\psi^{i},
\end{equation}
and every $\mathbf{u}_{i}$ fulfills $\mathbf{u}_{i}f=\mathbf{u}_{i}$.

The dual space of $S$, $S^{*}$, given by the set of $\mathbb{K}$-linear functions from $S$ to $\mathbb{K}$ is a left $\mathbb{K}$-module, which can be identified with $f\cl_{p,q}$. We will denote the dual basis of $\{\mathbf{u}_{i}\}$ by $\{\utilde{\mathbf{u}}^{i}\}$, implying:
\begin{align}
\utilde{\mathbf{u}}^{i}&\in{S}^{*}\\
\utilde{\mathbf{u}}^{i}(\mathbf{u}_{j})&=\tensor{\delta}{^{i}_{j}}f \ .
\end{align}
Every $\omega\in{S}^{*}$ can be written as:
\begin{equation}
\omega=\omega_{i}\utilde{\mathbf{u}}^{i},
\end{equation}
and the set
\begin{equation}
\left\{\mathbf{u}_{i}\utilde{\mathbf{u}}^{j}\ | \ i,j\in\{1,...,k\}\right\},
\end{equation}
is a basis of $\opn{End}_{\mathbb{K}}(S)\cong{\cl_{p,q}}$.

In \cite{VdR}, an interesting way to provide a $\K$-basis for $S$ is proposed. If we have a complete set of orthogonal primitive idempotents $f_{1}+...+f_{k}=1$, then we have the decomposition of the algebra given by Eq. \eqref{peirce}, and the existence of elements $\mathcal{E}_{ij}$ as in \eqref{baseconE}. If we take one of the primitive idempotents, for instance $f_{1}$, to generate the space of spinors, $S=\cl_{p,q}f_{1}$, then $\{\mathcal{E}_{i1}\ |\ i\in\{1,...,k\}\}$ constitutes a $\mathbb{K}$-basis for $S$, while $\{\mathcal{E}_{1i}\ |\ i\in\{1,...,k\}\}$ is a basis for $S^{*}=f_{1}\cl_{p,q}$.

Furthermore, if we start with a $\K$-basis $\{\ub_{i}\}$ for $S$; then for any fixed $i$ we have the element $\mathbf{u}_{i}\utilde{\mathbf{u}}^{i}$ is a primitive idempotent, and ${\sum}_{i=1}^{k}{\mathbf{u}_{i}\utilde{\mathbf{u}}^{i}}=1$. In this way, not only a complete set of orthogonal idempotents gives rise to a $\K$-basis for $S$, but also such a basis produces a complete set of orthogonal idempotents. Indeed it is possible to make the identification $\mathbf{u}_{i}\utilde{\mathbf{u}}^{j}=\mathcal{E}_{ij}$, and render the two approaches equivalent.

\subsection{Inner products on spinor spaces}

It is possible to define a \emph{conjugation} on the skewfield $\mathbb{K}$, and also define a inner product on the $\mathbb{K}$-module $S$. We will give the definition of inner product to be used in this article, since it will differ from the classical one on the positive-definiteness condition. From this section on we are going to use the definition $\K:=f\clpq{f}$, and in section \ref{Kdefn} we are going to comment on this and compare this choice with that of Eq. \eqref{divalgebradef}.

\begin{defn}
A \emph{conjugation}, $^{*}$, on a division ring $\mathbb{K}$ with unit $1_{\K}$ is an involution on $\mathbb{K}$ satisfying:
\begin{enumerate}
\item ${1_{\mathbb{K}}}^{*}=1_{\mathbb{K}}$,
\item $(ab)^{*}=b^{*}a^{*}$.
\end{enumerate}
\end{defn}

\begin{defn}\label{scalardef}
Let $S$ be a right module over a division ring $\mathbb{K}$, and $^{*}$ a conjugation on $\K$. We say that an application $h:S\times{S}\rightarrow\mathbb{K}$ is an \emph{inner product} or a \emph{scalar product} on S if the following properties are satisfied:
\begin{minipage}[t]{0.5\textwidth}
\begin{enumerate}
\item $h(\psi,\phi)=h(\phi,\psi)^{*}$,
\item $h(\psi,\phi+\chi)=h(\psi,\phi)+h(\psi,\chi) $,
\end{enumerate}
\end{minipage}%
\begin{minipage}[t]{0.5\textwidth}
\begin{enumerate}
\setcounter{enumi}{2}
\item $h(\psi,\phi\lambda)=h(\psi,\phi)\lambda \quad \forall{\lambda\in\mathbb{K}}$,
\item If $h(\psi,\phi)=0\quad\forall{\phi}$ then $\psi=0\in{S}$\ \ (non-degeneracy).
\end{enumerate}
\end{minipage}
\end{defn}

As customary, additional properties regarding this product are inferred from this axioms. For instance, that $h$ is anti-linear on the first argument: $h(\psi\lambda,\phi)=h(\phi,\psi\lambda)^{*}=(h(\phi,\psi)\lambda)^{*}=\lambda^{*}h(\psi,\phi)$.

Now we are in conditions to build different conjugations and inner products on spinor spaces of the Clifford algebra. In order to do so, we will make use of the two anti-automorphism of the Clifford algebra, namely reversion and Clifford conjugation.

It is possible to define an inner product in which the adjoint of elements in $\opn{End}_{\mathbb{K}}(S)$ coincides with one of the mentioned anti-morphisms. In this way, we will call $\bar{h}$ the inner product for Clifford conjugation, and $\hat{h}$ the one for reversion, which will have the properties:
\begin{align}\nonumber
\bar{h}(a\psi,\phi)&=\bar{h}(\psi,\bar{a}\phi),\\
\hat{h}(a\psi,\phi)&=\hat{h}(\psi,\hat{a}\phi),\label{inneradj}
\end{align}
where $\overline{a}$ denotes the Clifford conjugation applied on $a$, and $\hat{a}$ the reversion on $a$. Following reference \cite{VdR}, we will refer to any of the anti-morphisms by $^{\circ}$ and to the corresponding inner product simply by $h$. This two products are related by the \emph{structure spinor mapping}\cite{VdR}.

For every $\psi\in{S}$, $\psi{f}=\psi$, hence $\psi^{\circ}={f}^{\circ}\psi^{\circ}$. If $f^{\circ}=f$, then $\psi^{\circ}\psi\in{f\cl_{p,q}f}$ and the operation $\psi^{\circ}\psi$ defines an inner product satisfying relation (\ref{inneradj}). However if $f^{\circ}\neq{f}$, there's always an invertible element $s\in\cl_{p,q}$, such that $sf^{\circ}s^{-1}=f$\cite{Lounesto,VdR}\footnote{For simple algebras, this is a consequence of the adjoint action of $\clpq^{*}$ being transitive on primitive idempotents, together with $f$ and $f^{\circ}$ being both primitive idempotents. In semisimple algebras the adjoint action has two orbits on primitive idempotents, and $f$ and $f^{\circ}$ belong to the same one. See Ref. \cite{idempotentspaper} for a deeper analysis of idempotents on Clifford algebras}. Furthermore if the idempotent $f$ is built from elements of the standard basis that square to $1$, as it is customary and will be the case in section \ref{secembed}, $s$ can be chosen as an element from the standard basis. In this way the inner products defined by:
\begin{align}\nonumber
\hat{h}(\psi,\phi)=s\hat{\psi}\phi,\\
\overline{h}(\psi,\phi)=s\overline{\psi}\phi,\label{scalarprods}
\end{align}
with a suitable $s$ for each case, are scalar products. However, for the inner product to be well defined we still have to give the conjugation on $\mathbb{K}$.
\begin{defn}
Provided one of the antiautomorphisms of the Clifford algebra, $^{\circ}$, and $h(\psi,\phi)=s\psi^{\circ}\phi$, we define the \emph{conjugate} on an element $\lambda\in\mathbb{K}=f\cl_{p,q}f$, as
\begin{equation}\label{conjug}
\lambda^{*}=s\lambda^{\circ}s^{-1}.
\end{equation}
\end{defn}

Now, we can see that this operation defines a scalar product according to Def. \ref{scalardef}.

The group of transformation of spinors given by $\psi\mapsto{a\psi}$, where
\begin{equation}
a\in\opn{Aut}_{p,q}(S,h)=\{x\in\opn{\Gamma}_{p,q}\text{ such that }x^{\circ}x=1\},
\end{equation}
leave the inner product $h$ invariant. Hence, this group is usually referred to as \emph{the automorphism group of the inner product $h$}.

Lastly, note that since $h$ is non-degenerate, for every spinor $\psi$ there's a dual spinor $\psi^{\bullet}\in{S}^{*}$, such that $\psi^{\bullet}(\phi)=h(\psi,\phi)$ for every $\phi\in{S}$.

Since $$h(\mathbf{u}_{i}\psi^{i},\mathbf{u}_{j}\phi^{j})={\psi^{i}}^{*}h(\mathbf{u}_{i},\mathbf{u}_{j})\phi^{j}, $$every inner product can be represented by the hermitian matrix of components $h_{ij}=h(\mathbf{u}_{i},\mathbf{u}_{j})=s{\mathbf{u}_{i}}^{\circ}\mathbf{u}_{j}$, which is invertible because $h$ is non-degenerate. 

\section{Real spinors in $\cldostresr$, the inner product compatible with $\clunotres(\Bbb{C})$ and its Automorphism group}\label{secembed}

The set of algebraic spinors for $\cl_{2,3}(\mathbb{R})$ is determined once we pick a primitive idempotent $f$. In the complex Clifford algebra of Minkowski space using in physics, the Dirac representation of gamma matrices is achieved by taking the idempotent $f=1/4(1+\gamma_{0})(1+i\gamma_{1}\gamma_{2})$, (along with a suitable basis for $\cl_{1,3}(\mathbb{C})f$) which by means of the twisted embedding of $\clunotresr$ into $\cldostresr$ results in:
\begin{equation}\label{idempot}
f=\frac{1}{2}(1+\tilde{e_{0}})\frac{1}{2}(1+i\tilde{e_{1}}\tilde{e_{2}})=\frac{1}{2}(1+e_{1}e_{2}e_{3})\frac{1}{2}(1-e_{0}e_{3}e_{4}).
\end{equation}

The space of algebraic spinors will be the the left ideal $S=\cl_{2,3}f$. As it is well known, this ideal has a right linear structure over the division ring $f\cl_{2,3}f\cong{\mathbb{C}}$, which makes the space of spinor a complex vector space. Indeed, according to Eq. \eqref{centerKCl} we have:
\begin{equation}
f\cl_{2,3}f=\mathbb{R}f\oplus{i}\mathbb{R}f=\mathcal{Z}(\cl_{2,3})f\cong\Bbb{C}\ ,
\end{equation}
where $i=e_{0}e_{1}e_{2}e_{3}e_{4}$ is the (central) pseudoscalar of the algebra.

As was stated before, the isomorphism between $\clunotres(\mathbb{C})$ and $\cldostres(\mathbb{R})$, implies that the spinors of each of these algebras are exactly the same. Indeed, a real spinor in $\cldostres(\mathbb{R})$ is written as:
\begin{multline}
\psi=(\sigma+V^{A}e_{A}+S^{AB}e_{A}e_{B}+\\
T^{ABC}e_{A}e_{B}e_{C}+K^{ABCD}e_{A}e_{B}e_{C}e_{D}+i\omega)f,
\end{multline}
where $A,B,C$ and $D$ take values in $\{0,...,4\}$. Taking into account that $e_{4}f=-ie_{0}f$, then every term that has $e_{4}$ as a factor can be substituted by an imaginary component of an element in the Clifford algebra generated only by the first four unit vectors. Hence:
\begin{equation}
\cldostres(\mathbb{R})f=\clunotres{(\mathbb{R})}f+i\clunotres{(\mathbb{R})}f\cong\clunotres(\mathbb{C})f.
\end{equation}

A remarkable property of obtaining the complexified Clifford algebra in this manner is that the complex structure on the spinors is obtained from the free right $\mathbb{K}$-module structure that the spinor space have. Any algebraic spinor $\psi$ can be written as:
\begin{equation}
\mathbf{\psi}=\mathbf{u}_{i}\psi^{i},
\end{equation}

where $\{\mathbf{u}_{i}=\mathbf{u}_{i}f \}$ is a basis for $S$ as a right free module over $f\cl_{2,3}f$, and $\psi^{i}$ are the corresponding coefficients. Since $\cl_{2,3}\cong{\mathcal{M}(4,\mathbb{C})}$ then we have that $S\cong{\mathbb{C}^{4}}$. Using the command \texttt{spinorKbasis} from the \texttt{Clifford} package for Maple\cite{clipack} (as explained in appendix \ref{KbasisS}) we obtain a $\mathbb{K}$ basis for $S$,  $\{\mathbf{u_{i}}\}$:
\begin{equation}\label{base}
\mathbf{u}_{1}=f, \ \mathbf{u}_{2}=e_{1}f, \ \mathbf{u}_{3}=e_{0}f, \ \mathbf{u}_{4}=e_{0}e_{1}f .
\end{equation}

Once we have defined the algebraic spinor space with its right linear $\mathbb{C}$ structure, then it is possible to define the inner product associated to it.

Recall that any maximally symmetric inner product is related to some of the anti-automorphisms in the Clifford algebra, rather Clifford conjugation or reversion. Let us note that we want a conjugation operation which mimics the complex conjugation of the complex numbers. Let ``$^{-}$'' be the Clifford conjugation, it can be verified that:
\begin{equation}
\overline{f}=f, \ \ \overline{i}=-i.
\end{equation}

This implies that the inner product in spinor space associated to the Clifford conjugation takes the simple form:
\begin{equation}\label{prodclif}
\bar{h}(\psi,\phi)=\overline{\psi}\phi,
\end{equation} 
and also that the complex conjugation in $f\cl_{2,3}f$ is simply the Clifford conjugation. We can see from Eq. \eqref{prodclif} that the dual of $\psi$ coincides with its Clifford conjugated.
We can compute the matrix elements of this inner product $h_{ij}=\overline{\mathbf{u}_{i}}\mathbf{u}_{j}\in{\mathbb{K}}$ using the Clifford Algebra package for Maple and we obtain:
\begin{equation}\label{gamma0}
h_{ij}=\overline{\mathbf{u}_{i}}\mathbf{u}_{j}=
\begin{pmatrix}
f\ &0 \ &0 \ &0\\
0\ &f \ &0\ &0\\
0\ &0\ &-f \ &0\\
0\ &0\ &0 \ &-f							
\end{pmatrix}\ ,
\end{equation} 
which has the form of the matrix $\gamma_{0}$ in the Dirac representation, with the unit $1$ replaced by the unit $f$ of $f\cldostresr{f}$. Hence, the dual spinor in the dual basis, is given by
\begin{equation}
({\psi^{i}}^{*}\overline{\mathbf{u}_{i}}\mathbf{u}_{j})\utilde{\mathbf{u}}^{j}=(\psi^{\dagger}\gamma_{0})_{j}\utilde{\mathbf{u}}^{j},
\end{equation}
which coincides with the dual spinor in $\cl_{1,3}(\mathbb{C})$. Let us note that once we have obtained $\{\mathbf{u}_{i}\}$ it is easier to compute the elements $\{\overline{\mathbf{u}_{i}}\}$ than $\{\utilde{\mathbf{u}^{i}}\}$. It is clear however that $\utilde{\mathbf{u}^{j}}=h^{ji}\overline{\mathbf{u}_{i}}$, where $h^{ij}$ is the inverse matrix of $h_{ij}$. In our particular case $h^{ij}=h_{ij}$.

The group of automorphisms of this inner product is given by:
\begin{equation}
\opn{Aut}_{2,3}(S,\overline{h})=\{x\in\opn{\Gamma}_{2,3}\ | \ \overline{x}x=\mathbf{1}\},
\end{equation}
however, in order to consider only transformations that are closed in ${\mathbb{R}^{1,3}}$ it would be more appropriate to restrict ourselves to the group:
\begin{equation}
^{5D}\opn{Aut}_{1,3}(S,\overline{h})=\{x\in\opn{Aut}_{2,3}(S,\overline{h}) \ | \ xvx^{-1}\in\Bbb{R}^{1,3},\ \ \forall{v\in\mathbb{R}^{1,3}}\ \}.
\end{equation}
Clearly  $\opn{Aut}_{1,3}(S,\overline{h})=\{x\in\opn{\Gamma}_{1,3}\ | \ \overline{x}x=\mathbf{1}\}=\tensor{\opn{Pin}}{^{\wedge}_{+}}(1,3)$ is contained in this group, but $^{5D}\opn{Aut}_{1,3}(S,\overline{h})$ is larger. For instance the unit circle $S^{1}$ contained in the center of $\cldostresr$ is in $^{5D}\opn{Aut}_{1,3}(S,\overline{h})$ but not in $\opn{Aut}_{1,3}(S,\overline{h})$. This fact allows us to treat the electromagnetic $\opn{U}(1)$-gauge transformation from the perspective of real Clifford algebras\cite{arcodiaparantipar}.

Any element $A\in\cldostres$ can be written represented in the basis $\{\mathbf{u}_{i}\utilde{\mathbf{u}}^{j}\}$ by the components $\tensor{A}{^{i}_{j}}$ defined as\footnote{See section \ref{Kdefn} to see in which sense $A^{i}_{j}$ are the components of $A$ in the basis $\{\mathbf{u}_{i}\utilde{\mathbf{u}}^{j}\}$}:
\begin{equation}
\tensor{A}{^{i}_{j}}=\utilde{\mathbf{u}}^{i}A\mathbf{u}_{j}=h^{il}\overline{\mathbf{u}_{l}}A\mathbf{u}_{j}.
\end{equation}
Using this definition together with the facts that $\tilde{e}_{0}\mathbf{u}_{1}=\mathbf{u}_{1}$, $\tilde{e}_{0}\mathbf{u}_{2}=\mathbf{u}_{2}$, $\tilde{e}_{0}\mathbf{u}_{3}=-\mathbf{u}_{3}$ and $\tilde{e}_{0}\mathbf{u}_{4}=-\mathbf{u}_{4}$; it can be checked that the components of $\gamma_{0}$, that we have identified with $\tilde{e}_{0}$, are given by $\tensor{(\gamma_{0})}{^{i}_{j}}=\tensor{(\tilde{e}_{0})}{^{i}_{j}}=h_{ij}=h^{ij}$, in consistence with Eq. \eqref{gamma0}.

If we define the \emph{matrix adjoint of an element} $A$, $A^{\dagger}$, as the element whose components are:
\begin{equation}
\tensor{(A^{\dagger})}{^{i}_{j}}=(\tensor{A}{^{j}_{i}})^*,
\end{equation}
we have the relation
\begin{equation}\label{matadj}
A^{\dagger}=\gamma_{0}\overline{A}\gamma_{0},
\end{equation}
which can be proven component-wise by direct calculation:
\begin{align*}
\tensor{(A^{\dagger})}{^{i}_{j}}&=(\tensor{A}{^{j}_{i}})^*=\overline{\tensor{A}{^{j}_{i}}}=\overline{h^{jl}\overline{\mathbf{u}_{l}}A\mathbf{u}_{i}}=\overline{\ub_{i}}{}\overline{A}u_{l}\overline{h^{jl}}=h_{im}\utilde{\ub^{m}}{}\overline{A}{}u_{l}{h^{lj}}=&\\
&=h_{im}\tensor{\overline{A}}{^{m}_{l}}{h^{lj}}=\tensor{\left(\gamma_{0}\right)}{^{i}_{m}}\tensor{\overline{A}}{^{m}_{l}}\tensor{\left(\gamma_{0}\right)}{^{l}_{j}}=\tensor{\left(\gamma_{0}\overline{A}\gamma_{0}\right)}{^{i}_{j}}.&
\end{align*}
Note that, because $\overline{\tilde{e_{\mu}}}=\tilde{e_{\mu}}$, and we have identified $\gamma_{\mu}$ with $\tilde{e_{\mu}}$, then Eq. \eqref{matadj} implies the known property:
\begin{equation}
{\gamma_{\mu}}^{\dagger}=\gamma_{0}\gamma_{\mu}\gamma_{0}.
\end{equation}

We have yet to analyze the inner product associated to reversion. Since $\hat{f}\neq{f}$ it is necessary to find an element $s$, such that $s\hat{f}s^{-1}=f$. It can be verified that $s=e_{0}e_{1}$ is such an element\footnote{In appendix \ref{sforreversion} we write a very simple algorithm to find such an element in the Clifford standard basis.}. Using Eq. \eqref{conjug} to define the conjugation in $f\clpq{f}=\mathcal{Z}(\clpq)f$, we obtain that:
\begin{equation}
f^{*}=f \ \ ; \ \ (if)^{*}=if \ \ ,
\end{equation}  
which implies that the conjugation of scalars for reversion is simply the identity. We conclude that this inner product doesn't reproduce the usual complex structure of $\clunotres(\mathbb{C})$, and don't go further in its study.

\subsection{On the possible definitions of $\K$}\label{Kdefn}

From section \ref{spi} on, we have taken the definition of $\K$ to be $f\cldostresr{f}$, instead of the more well behaved definition in Eq. \eqref{divalgebradef} that needs a whole set of orthogonal primitive idempotents. There are some aspects that need to be clarified regarding the set $\mathcal{B}=\left\{\mathbf{u}_{i}\utilde{\mathbf{u}}^{j}\ | \ i,j\in\{1,...,k\}\right\}$ being a $\K$-basis for $\clpq$.

First suppose that we had defined $\K$ to be $f\clpq{f}$, and we wanted to expand an element $a\in\clpq$ in the mentioned $\K$-basis. In this framework, we would have $A=A^{i}_{j}\mathbf{u}_{i}\utilde{\mathbf{u}}^{j}$, but since $A^{i}_{j}\in{f}\clpq{f}$ this would imply that $A\in{f}\clpq\neq\clpq$ and consequently, not every element in the Clifford algebra could be represented in this way. This would imply that $\mathcal{B}$ is not a basis. What is really happening is that we are using either the wrong scalar multiplication on $\mathcal{B}$, or the wrong definition of $\K$. If we want to keep the product to be the usual Clifford product we need to consider the definition in Eq. \eqref{divalgebradef} for $\K$. In this sense, note that $A$ can be written as $A^{i}_{j}\mathbf{u}_{i}\utilde{\mathbf{u}}^{j}$ where $A^{i}_{j}=\ub_{k}\utilde{\ub^{i}}A\ub_{j}\utilde{\ub}^{k}\in\K$:
\begin{equation}
A^{i}_{j}\ub_{i}\utilde{\ub}^{j}=\ub_{k}\utilde{\ub^{i}}A\ub_{j}\utilde{\ub}^{k}\ub_{i}\utilde{\ub}^{j}=\ub_{k}\utilde{\ub^{i}}A\ub_{j}\delta^{k}_{i}\utilde{\ub}^{j}=\underbrace{\ub_{i}\utilde{\ub^{i}}}_{=\mathbf{1}}A\underbrace{\ub_{j}\utilde{\ub}^{j}}_{=\mathbf{1}}=A,
\end{equation}
and the same can be proved for $\mathbf{u}_{i}\utilde{\mathbf{u}}^{j}A^{i}_{j}$. However, there is an isomorphism between $\K$ and $f\clpq{f}$ given by:
\begin{center}
\begin{minipage}{0.4\textwidth}
\begin{align} \nonumber
\varphi:\K&\rightarrow{f}\clpq{f},\\ \nonumber
z&\mapsto{z}f,
\end{align}
\end{minipage}
\hspace{5mm}
\begin{minipage}{0.4\textwidth}
\begin{align} \nonumber
\varphi^{-1}:f\clpq{f}&\longrightarrow\K, \\ \nonumber
\lambda&\longmapsto\ub_{k}\lambda\utilde{\ub}^{k},
\end{align}
\end{minipage}
\end{center}
which allows us to represent any element of the CA using coefficients in $f\clpq{f}$ instead of $\K$. By inspecting the form of the coefficients $A^{i}_{j}\in\K$, we see that the corresponding elements in $f\clpq{f}$ by $\varphi$ is given by $\utilde{\ub}^{i}A\ub_{j}$, and we have that $A=\varphi^{-1}(\utilde{\ub}^{i}A\ub_{j})\ub_{i}\utilde{\ub}^{j}$ (which is indeed also a scalar product). In this sense, there is enough information in the elements $\utilde{\ub}^{i}A\ub_{j}\in{f}\clpq{f}$, and we say that this elements represent $A$ in the basis $\{\ub_{i}\utilde{\ub}^{j}\}$.

An important consequence of the choice $\K=f\clpq{f}$ is that the inner products take the very simple forms \eqref{scalarprods}. Otherwise, they should be mapped to $\K$ by means of the isomorphism $\varphi^{-1}$, resulting in the redefinition of the scalar products:
\begin{align}\nonumber
\hat{h}(\psi,\phi)=\ub_{k}\left(s\hat{\psi}\phi\right)\utilde{\ub}^{k},\\
\overline{h}(\psi,\phi)=\ub_{k}\left(s\overline{\psi}\phi\right)\utilde{\ub}^{k},
\end{align}
for each antimorphism. Hence, the more appropriate definition of $\K$ depends on the problem we are working on.

\section{Final comments}\label{secfinal}

We have studied the structure of the Clifford algebra $\cldostresr$ and its relation with the complex spacetime algebra. On the first part of this article we analyzed the embedding of $\clunotresr\subseteq\cldostresr$, and established the inclusion $\opn{Pin}(1,3)\hookrightarrow{\opn{Spin}}(2,3)$. In order to do so, we used the fact that for odd-dimensional spaces, the adjoint action only covers the special orthogonal group, and that there exist an injection $\opn{O}(1,3)\hookrightarrow\opn{SO}(2,3)$. In particular we have obtained the elements in $\opn{Spin}(2,3)$ that correspond to time reversal and parity.

We have also studied two possible ways of embedding the Minkowski vector space into $\cldostresr$: the trivial one $\gamma_{\mu}\mapsto{e}_{\mu}$, whose image we have called $\mathbb{R}^{1,3}$; and the twisted embedding $\gamma_{\mu}\mapsto-ie_{4}e_{\mu}$, whose image we have called $\widetilde{\Bbb{R}^{1,3}}$. Regarding these spaces, we have shown that the adjoint action of $\pinunotres$ is well defined not only for $\R^{1,3}$, but also for $\widetilde{\R^{1,3}}$. Furthermore the adjoint action of an element $x\in\pinunotres$ on the element $v^{\mu}\tilde{e_{\mu}}$ is equal to the \emph{twisted} adjoint action of the same element $x$ on the vector $v^{\mu}e_{\mu}$. 

On the second part of this article we have studied the real spinors of $\cldostresr$. In this analysis we have defined the complex structure of the algebra which renders $\clunotresc\cong\cldostresr$, and constructed the inner product of spinors associated with Clifford conjugation. We have also proven that this product is the one compatible with the usual product of complex spinors in $\clunotresr$. Additionally we have defined the group $^{5D}\opn{Ad}_{1,3}(S,\overline{h})$ that consist of the elements leaving the spinor inner product invariant and at the same time preserving the spaces $\mathbb{R}^{1,3}$ and $\widetilde{\mathbb{R}^{1,3}}$ under the adjoint action.

In this second part we have made use of the \texttt{Clifford} package for Maple\cite{clipack} to compute a $\mathbb{K}$-basis for the spinor space and also to work with the reversion inner product.

From the point of view of physics, using $\cldostresr$ to complexify the real algebra $\clunotresr$ implies adding a timelike extra dimension, and working in a universe with two times\cite{bars}. We recently used this approach to provide an interpretation of particle and antiparticle states in the Dirac theory\cite{arcodiaparantipar}, and a more detailed and rigorous study of the Clifford algebra $\cldostresr$ and the compatibility with the embedding of the Pin group from a mathematical point of view is given in this article.

Note that at an algebraic level there is not distinction between usual complexification and adding an extra timelike dimension, hence, results involving complex numbers can also be explained by an extra timelike dimension. However, the hypothesis of the extra timelike dimension could be ruled out in the context of induced matter theories for the massive spinor fields. The starting point of these theories\cite{arcodiaparantipar,sanchez,wesson} is a massless Dirac equation in the 5D spacetime, from which a massive 4D Dirac equation is obtained once the fifth component of the momentum vector of the field is identified as the induced mass in the 4D spacetime. Consequently the mass is not a scalar but the last component of a 5-vector. This implies that the parity and time reversal operators, $\mathcal{P}$ and $\mathcal{T}$, obtained in the present article would reverse the sign of the mass. In a future work it would be interesting to study this possibility, and analyze if there is a relation with the work of Barut and Ziino\cite{barut}, who considered Dirac equations with opposite sign of mass for particles and antiparticles.

An interesting problem to study in the future is if the two times complexification of the Clifford algebra can be extended to \emph{quantum Clifford algebras} and \emph{quantum algebraic spinors}\cite{quantumspinors,quantumfauser,quantumlounesto}. Quantum Clifford algebras are obtained when an antisymmetric part is added to the bilinear form $g_{\mu\nu}dx^{\mu}\otimes{dx^{\nu}}$. In this fashion $g_{\mu\nu}$ is replaced by $g_{\mu\nu}+A_{\mu\nu}$ where $A_{\mu\nu}$ is antisymmetric. These algebras are important in high energy regimes as they account for quantum phenomena \cite{quantumspinors,quantumfauser}. Quantum algebraic spinors are defined in the same fashion as algebraic spinors but when the Clifford algebra is replaced by the quantum Clifford algebra. As can be seen in Ref. \cite{quantumspinors} the idempotent $f$ in Eq. \eqref{idempot} changes according to:
\begin{equation}
f\mapsto{f}+f(A)=f+\frac{i}{4}(A_{12}\gamma_{0}+A_{20}\gamma_{1}+A_{01}\gamma_{2}),
\end{equation}
where $f(A)$ is an imaginary vector in 4D. However in our 5D approach, by means of the twisted embedding of $\clunotresr$ into $\cldostresr$, we can see that $f(A)$ would be a real bivector instead. In general an arbitrary element of the Clifford algebra, $b$, changes according to:
\begin{equation}
b\mapsto{b} + b(A),
\end{equation}
where in the usual $\clunotresc$ formalism the component $b(A)$ is an element in $\bigwedge_{0}(V)\oplus\bigwedge_{1}(V)\oplus\bigwedge_{2}(V)$ but in our 5D approach it is an element in $\bigwedge_{0}(V)\oplus\bigwedge_{3}(V)\oplus\bigwedge_{2}(V)$. 

Since extra-dimensions are widely used in theories of unification and quantum gravity, and since in particular 5D spacetimes are considered the low-energy limit of this theories (e.g. 11D supergravity\cite{sugra,sugra5d}) it would be important to extend the two times complexification of the spacetime algebra to the quantum spacetime algebra.
 
\appendix
\section{Code for computations on $\cldostres(\mathbb{R})$ using the \texttt{Clifford} package for Maple}
In this appendix we use the \texttt{Clifford} package for Maple\cite{clipack}, developed by Ab{\l}amowicz and Fauser to perform some calculations.
We are going to use the following preamble, which contains information about the Clifford algebra and define the idempotent $f$, for all the calculations:
\begin{center}
\begin{BVerbatim}
with(Clifford):with(LinearAlgebra):
clibasis:=cbasis(5);
B:=linalg[diag](1,-1,-1,-1,1);
f:=cmul(1/2*(1+e2we3we4),1/2*(1-e1we4we5));
\end{BVerbatim}
\end{center}

\subsection{Computation of the $\mathbb{K}$-basis for $S$}\label{KbasisS}

In order to compute the $\mathbb{K}$-basis for $S$ we use the command \texttt{spinorKbasis}. This command takes four inputs: 
\begin{enumerate}
\item a real basis for the space $S$,
\item the idempotent $f$,
\item a list of elements from the standard basis of $\clpq$ that generate $\mathbb{K}$ as a real vector space.
\item the string ``left'' or ``right'' indicating if $S$ is the left minimal ideal, $\clpq{f}$, or the right one, $f\clpq$.
\end{enumerate}

As a matter of fact \texttt{spinorKbasis} returns three elements, the $\mathbb{K}$-basis for $S$ is the first one. Let's find all this objects needed as inputs using the Clifford package. To generate a real basis of $S$ we are going to use the command \texttt{minimalideal} in the following way:
\begin{center}
\begin{BVerbatim}
realbasisS:=minimalideal(clibasis,f,'left');
\end{BVerbatim}
\end{center}
Lastly, we define the set of generators of $\mathbb{K}$, and compute the $\mathbb{K}$-basis for $S$:
\begin{center}
\begin{BVerbatim}
Kgenerators:=[Id,e1we2we3we4we5];
KbasisS:=spinorKbasis(realbasisS,f,Kgenerators,'left')[1];
\end{BVerbatim}
\end{center}
The basis obtained in this way, although rendering $\gamma_{0}=-ie_{4}e_{0}$ diagonal, it is not the matrix for the Dirac representation. However, by the following reordering:
\begin{center}
\begin{BVerbatim}
DiracKbasisS:=[KbasisS[1],KbasisS[3],KbasisS[2],KbasisS[4]];
\end{BVerbatim}
\end{center}
we obtain the basis in Eq. \eqref{base}.

\subsection{Finding the element $s$ for the reversion inner product}\label{sforreversion}

The following algorithm finds the element $s$ in Eq. \eqref{conjug}, as long as it belongs in the standard basis, which is the case for the idempotent $f$ in Eq. \eqref{idempot}.
\begin{center}
\begin{BVerbatim}
l:=nops(clibasis): s:=0:
for i from 1 to l do
      if cmul(clibasis[i],reversion(f))-cmul(f,clibasis[i])=0 and s=0 
         then s:=clibasis[i]; 
      end if;
end do;
print(s);
\end{BVerbatim}
\end{center}
\bibliography{refs}
\bibliographystyle{spmpsci}

\end{document}